\documentclass[11pt]{amsart}
\usepackage{verbatim}
\usepackage{amsmath,amsfonts,amscd,amssymb,epsfig,euscript,bm}
\usepackage{amsxtra,mathrsfs}
\usepackage{pdfsync}
\newtheorem{theorem}{Theorem}[section]

\newtheorem{lemma}{Lemma}[section]

\theoremstyle{definition}
{}

\theoremstyle{remark} 
\newtheorem{remark}{Remark}[section]

\newcommand{\spbp}{\mathcal{\sqrt{-1}\partial \bar{\partial}}}

\numberwithin{equation}{section}
\usepackage{hyperref}
\hypersetup{colorlinks,citecolor=blue,plainpages=false,hypertexnames=false}
\begin{document}
\title[Quillen metrics and perturbed equations]{Quillen metrics and perturbed equations}
\author{Vamsi Pritham Pingali}
\address{Department of Mathematics, Indian Institute of Science, Bangalore, India - 560012}
\email{vamsipingali@math.iisc.ernet.in}
\begin{abstract}
We come up with infinite-dimensional prequantum line bundles and moment map interpretations of three different sets of equations - the generalised Monge-Amp\`ere equation, the almost Hitchin system, and the Calabi-Yang-Mills equations. These are all perturbations of already existing equations. Our construction for the generalised Monge-Amp\`ere equation is conditioned on a conjecture from algebraic geometry. In addition, we prove that for small values of the perturbation parameters, some of these equations have solutions. 
\end{abstract}
\maketitle
\section{Introduction}\label{Introsec}
\indent The Kempf-Ness theorem of algebraic geometry roughly states that if a Lie group acts on a K\"ahler manifold in a Hamiltonian manner, then in every gauge orbit of a stable point there exists a zero of the moment map. Surprisingly, many PDE of mathematical physics and geometric analysis can be formally interpreted as the zeroes of infinite-dimensional moment maps of infinite-dimensional symplectic manifolds. Thus the finite-dimensional Kempf-Ness theorem can guide us to find the correct algebro-geometric stability obstruction towards solving the PDE. But if we want to be faithful to the finite-dimensional Kempf-Ness theorem, we would want the infinite-dimensional symplectic form to arise as the curvature of an infinite-dimensional hermitian holomorphic line bundle. Indeed, the correct infinite dimensional line bundle often turns out to be a Quillen determinant bundle (see \cite{don2} for instance). \\
\indent From a physics perspective, many infinite-dimensional symplectic manifolds that arise in the above picture are spaces of fields like Gauge fields, Higgs fields, etc. To quantise these fields, i.e., to do quantum field theory, one would like to formally construct a prequantum line bundle, the space of sections of which would form the prequantum Hilbert/Fock space. So for multiple reasons it is beneficial to identify not only a moment map interpretation of a given PDE, but also the prequantum line bundle whose curvature is the correct symplectic form.\\
\indent In this paper we identify the prequantum line bundles for three sets of equations - the generalised Monge-Amp\`ere equation, the almost Hitchin system, and the Calabi-Yang-Mills equations. All of these are in some sense, perturbations of already existing equations, namely, the usual Monge-Amp\`ere equation, the Hitchin system, and the decoupled system of the Hermite-Einstein and the Monge-Amp\`ere equation respectively. Interestingly enough, our construction of the prequantum bundle for the generalised Monge-Amp\`ere equation is conditioned on the Hodge conjecture. With these generalities in mind, we proceed to describe the main results of the paper.\\

\indent Let $(M,\omega)$ be a compact K\"ahler $n$-complex dimensional manifold such that $[\omega]=[c_1(L,h)]$ for some hermitian holomorphic line bundle $(L,h)$ satisfying $\displaystyle \int \omega^n =1$. Our first result concerns the generalised Monge-Amp\`ere equation.
\begin{theorem}  For $1\leq k \leq n$ let $\alpha_k$ be closed real $(k,k)$-forms such that $[\alpha_k] \in H^{2k}(M,\mathbb{Q})$, $\alpha_n >0$, $\displaystyle \int _M \omega^n = \int _M \sum_k \alpha_k \omega^{n-k}$, and $n\omega^{n-1} - \sum_{k} (n-k)\omega ^{n-k-1} \alpha_k >0$. Suppose $\mathcal{A}^{1,1}$ is the space of smooth unitary integrable connections on $L$.  \\
Assuming that $[\alpha_k] = [\mathrm{ch}(\mathcal{E}_k)]$ for some holomorphic virtual bundles $\mathcal{E}_k$, there exists a holomorphic line bundle $\mathbf{Q}$ on $\mathcal{A}^{1,1}$ with a unitary Chern connection whose curvature $\Omega$ is a symplectic form on $\mathcal{A}^{1,1}$. Moreover, the unitary gauge group $\mathcal{G}$ acts in a Hamiltonian manner on $\mathcal{A}^{1,1}$ with a moment map $\mu$. There is a zero of the moment map in every complex gauge orbit if and only if the following equation is satisfied for a smooth function $\phi$. (In what follows, $\omega_{\phi}=\omega+\spbp \phi$.) 
\begin{gather}
\omega_{\phi}^n = \displaystyle \sum \alpha _k \omega_{\phi}^{n-k}.
\label{genma}
\end{gather}
\label{maingenma}
\end{theorem}
We note the easy fact that for all $1\leq k\leq n-1$, for small $\alpha_k$ the generalised Monge-Amp\`ere equation has a solution \cite{Pingen}. \\
\begin{remark}
Note that theorem \ref{maingenma} is actually assuming the Hodge conjecture because the forms $\alpha_k$ are in arbitrary rational cohomology classes.
\end{remark}

\indent The next result produces a new almost Hitchin system of equations (depending on a positive integer $k$) akin to the almost Hermite-Einstein equation of Leung \cite{Leung}. We call solutions of the almost Hitchin equations, almost Higgs connections. The whole point is to give a differential geometric interpretation of Higgs-Gieseker stability which was recently studied in \cite{Gieskerhiggs}. Unfortunately, so far we have been able to give a complete proof only for the existence of approximate almost Higgs connections for Higgs-Mumford semistable bundles.
\begin{theorem}
Let $(E,h_E)$ be a holomorphic hermitian vector bundle. Suppose $\mathcal{A_E}^{1,1}$ is the space of smooth unitary integrable connections on $E$ and $\mathcal{B}=\Omega^{1,0} (End(E))$ is the space of endomorphism-valued $(1,0)$-forms. Assume that we are given a holomorphic $\Phi \in \mathcal{B}$ which satisfies $\Phi \wedge \Phi =0$. Fix an integer $l$.  \\
\indent There exists three holomorphic hermitian line bundles $\mathbf{Q_{k,l}}$, $\mathbf{H}_1$, and $\mathbf{H}_2$ on $\mathcal{A_E}^{1,1}\times \mathcal{B}$ depending on $E$ and $L^k$  with a unitary Chern connection whose curvatures $\Omega, \tilde{\Omega}_1$ and $\tilde{\Omega}_2$ are symplectic forms on $\mathcal{A_E}^{1,1}\times\mathcal{B}$. Moreover, the unitary gauge group $\mathcal{G}$ acts in a Hamiltonian manner on $\mathcal{A_E}^{1,1}\times \mathcal{B}$ with moment maps $\mu_{k,l}, \tilde{\mu}_1$ and $\tilde{\mu}_2$. There is a simultaneous zero of both moment maps in every complex gauge orbit if and only if the following system of equations is satisfied for a unitary connection $A$ giving an isomorphic complex structure on $E$ (via a gauge isomorphism $g$). 
\begin{gather}
[e^{k\omega+\frac{\sqrt{-1}}{2\pi}\Theta_A }Td(M,k\omega)]^{n,n} + n\frac{\sqrt{-1}}{2\pi}[\Phi,\Phi^{\dag}]\wedge (k+l)^{n-1} \omega^{n-1} = c_{k,l}\omega^n, \nonumber \\
d_A ^{0,1} g. \Phi =0 
\label{higgseq}
\end{gather}
where $c_{k,l}$ is a topological constant, $g. \Phi$ is the gauge transformation $g$ acting on $\Phi$ via the adjoint action, and $Td(M,k\omega)$ is the Todd form of the K\"ahler  metric $k\omega$.\\
\indent If the system \ref{higgseq} admits a smooth solutions $A_{k,l}$ for all sufficiently large $k>>1$ whose $C^a$ norms are bounded above uniformly in $k$ (for every fixed $a$), then the bundle $(E,\Phi)$ is Gieseker-Higgs stable.\\
\indent If $(E,\Phi)$ is irreducible (it is not a direct sum of $\Phi$-invariant holomorphic subbundles) and Higgs-Mumford semistable, then it admits an approximate solution to \ref{higgseq}, i.e., a connection that solves the following equations, for sufficiently large $k$ (that can potentially depend on $l$).
\begin{gather}
\Vert [e^{k\omega+\frac{\sqrt{-1}}{2\pi}\Theta_A }Td(M,k\omega)]^{n,n} + n\frac{\sqrt{-1}}{2\pi}[\Phi,\Phi^{\dag}]\wedge (k+l)^{n-1} \omega^{n-1} - c_{k,l}\omega^n\Vert < \epsilon, \nonumber \\
d_A ^{0,1} g. \Phi =0.
\label{approxhiggseq}
\end{gather}
\label{approxhiggs}
\end{theorem}
\begin{remark}
Actually theorem \ref{approxhiggs} is equivalent to the special case when $l=0$. However, we chose to state it this way to emphasize that $(E,\Phi)$ is Gieseker-Higgs stable if and only if $(E\otimes L^l, \Phi)$ is Gieseker-Higgs stable.
\label{useless}
\end{remark}

\indent The final result is a moment map interpretation of the Calabi-Yang-Mills equations introduced by the author \cite{Pinchern} as a toy model of the more complicated K\"ahler-Yang-Mills equations (whose point is to uniformise triples $(X,E,L)$ of manifolds with vector bundles) introduced and studied in \cite{Garcia1, Garcia2, Garcia3}. In \cite{Pinchern} they were introduced from an analytic point of view in an admittedly ad hoc manner. The following result motivates them as the zero locus of the moment map associated to a symplectic form arising as the curvature of prequantum line bundle.
\begin{theorem}
Let $\mathcal{A}_{\mathcal{E}}^{1,1}$ and $\mathcal{A}^{1,1}$ be the same as in theorems \ref{maingenma} and \ref{approxhiggs}. Let $\mathcal{N} = \mathcal{A}_{\mathcal{E}}^{1,1} \times  \mathcal{A}^{1,1}$. Consider the gauge group $\mathcal{G}=\mathcal{G}_E \times \mathcal{G}_L$ of unitary transformations of $E$ and $L$.\\
\indent Given a rational number $\alpha>0$, there exists a hermitian holomorphic line bundle $\mathbb{Q_{\alpha}}$ on $\mathcal{N}$ whose curvature $\Omega_{\alpha}$ is a symplectic form. Moreover, the group $\mathcal{G}$ acts on $\mathcal{N}$ in a Hamiltonian manner with a moment map $\mu_{\alpha}$. There is a zero of the moment map in every complex gauge orbit if and only if the following system of equations (the Calabi-Yang-Mills equations) is satisfied by an appropriate $(A,\phi)$. 
\begin{gather}
\sqrt{-1} \Theta_A \wedge n\omega_{\phi}^{n-1} = -\lambda \omega_{\phi}^n Id \nonumber \\
\omega_{\phi}^n \left (1+\frac{\alpha\lambda ^2 r}{2} \right )-\eta = \alpha \mathrm{ch}_2(A) n(n-1)\omega_{\phi}^{n-2},
\label{cym}
\end{gather}
where $\lambda$ is a topological constant, and $\eta$ is an $(n,n)$-form.  Moreover, if $\eta>0$, $\alpha$ is small, $E$ is Mumford-stable with respect to $L$, then the Calabi-Yang-Mills equations have a smooth solution.  
\label{calabiyangmills}
\end{theorem}
\begin{remark}
Normally, it is ``easy" to come up with a naturally occurring symplectic form whose moment map is the given PDE. It is harder to come up with the correct prequantum line bundle. However, in the case of theorem \ref{calabiyangmills}, interestingly enough, we come up with the correct symplectic form as the curvature of a very natural prequantum line bundle.
\label{strangerema}
\end{remark}
\emph{Acknowledgements} : Thanks is in order to Harish Seshadri for useful discussions. The author is also grateful to Joel Fine, Rukmini Dey, Leon Takhtajan, Richard Wentworth, Mario Garcia-Fernandez, and Indranil Biswas for answering some questions. The author acknowledges the support of an ECRA grant from SERB - ECR/2016/001356.
\section{The generalised Monge-Amp\`ere equation}\label{MAsec}
\indent In this section we prove theorem \ref{maingenma}. We adopt the same notation as in the statement of theorem \ref{maingenma}. \\
\indent The following setup is inspired by a similar one for the Calabi conjecture by Fine \cite{fine}. The space of unitary connections $\mathcal{A}$ is an affine space and is the same as the space of $\bar{\partial}$ operators which is actually a complex vector space. The integrable ones $\mathcal{A}^{1,1}$ form a complex submanifold of this manifold with tangent space at $A$ consisting of $(0,1)$-forms $a^{0,1}$ satisfying $\bar{\partial} a^{0,1}= 0$. It is clear that the gauge group $\mathcal{G}$ preserves integrability. The Lie algebra $T_I \mathcal{G}$ of the gauge group consists of imaginary functions $\sqrt{-1}H$. The bundle $\mathbf{Q}$ will be the Quillen determinant bundle (endowed with the corresponding Quillen metric) of a virtual bundle $\mathcal{E}$ on $M \times \mathcal{A}^{1,1}$. \\
\indent Firstly we prove a lemma that tells us what the ``correct" symplectic form $\Omega$ ought to be in order to get the generalised Monge-Amp\`ere equation.  
\begin{lemma}
The moment $\mu$ corresponding to the symplectic form
\begin{gather} 
\Omega_A(a,b) = W\displaystyle \int _M a \wedge b \wedge (n\omega_A^{n-1}-\sum_k (n-k)\alpha _k \omega_A^{n-k-1})
\label{symplecticform1}
\end{gather}
 is given by the following equation (where $W$ is any integer $\geq 1$).
\begin{gather}
\mu_A (\sqrt{-1}H) = W\displaystyle \int _M \sqrt{-1}H (\omega_A^n - \sum _k \alpha_k \omega_A^{n-k}). 
\label{mom}
\end{gather}
where $\omega_A$ is the first Chern form of the connection $A$.
\label{symplemm1}
\end{lemma} 
\begin{proof}
Let $b$ be an imaginary one-form. By definition of the moment map, the variation of $\mu$ around $A$ along the tangent vector $b$ ought to be equal to $-\Omega (\sqrt{-1}dH,b)$. Indeed,
\begin{gather}
\delta_b \mu_A (\sqrt{-1}H) = N\displaystyle \int_M \sqrt{-1} H(n\omega^{n-1} db-\sum_k (n-k)\alpha_k \omega^{n-k-1} db)   \nonumber \\
= -N\displaystyle \int_M \sqrt{-1}dH(n\omega^{n-1} b-\sum_k (n-k)\alpha_k \omega^{n-k-1} b) = -\Omega(\sqrt{-1}dH,b). \nonumber
\end{gather}
\end{proof}
\indent The complexification $\mathcal{G}_{\mathbb{C}}$ (which is simply $\mathbb{C}^{\*}$-valued functions on $M$) acts on the $\bar{\partial}$ operators according to $g. A^{0,1} = A^{0,1}+\bar{\partial} \ln (g)$. Therefore the first Chern form changes to $\omega+\sqrt{-1}\partial \bar{\partial} \ln (\vert g\vert^2) = \omega_{\phi}$ where $\phi = \ln (\vert g\vert^2)$. Thus, satisfying the generalised Monge-Amp\`ere equation in a given K\"ahler class is equivalent to finding a zero of the moment map in a complex orbit. \\
\indent Before proceeding further, we prove the following amusing lemma. It implies that the usual Hodge conjecture is equivalent to a potentially harder version which is, ``Given a form in a Hodge class, is it a rational linear combination of Chern-Weil forms of hermitian holomorphic vector bundles endowed with Chern connections ?"
\begin{lemma}
If $\alpha$ is a $(k,k)$ real closed form on $M$ whose cohomology class $[\alpha] = [ch(E)]$ where $E$ is a holomorphic vector bundle and $ch$ is the Chern character class, then $\alpha = ch(\mathcal{V})$ at the level of forms, where $\mathcal{V}$ is a virtual hermitian holomorphic vector bundle with a Chern connection.
\label{Hodge}
\end{lemma}
\begin{proof}
Choose some metric $h_0$ for $E$. Then $\alpha - ch_k(E,h) = \spbp \eta$ for some real form $\eta$. By the holomorphic Venice lemma (Theorem 1 in \cite{PiTa}) it follows that $\spbp \eta = ch(\mathcal{W})$ for some virtual hermitian holomorphic bundle $\mathcal{W}$. Therefore $\mathcal{V}=\mathcal{W}\oplus(E,h)$ does the job. 
\end{proof}
\begin{remark}
Recall the well known fact that the Lefschetz theorem on $(1,1)$-forms and the Hard Lefschetz theorem imply that every $(n-1,n-1)$ Hodge class $[\alpha]$ is of the form $N[\alpha] = [ch(\mathcal{V})]$ for some large integer $N$ and a virtual holomorphic bundle $\mathcal{V}$. By lemma \ref{Hodge} this holds true even at the level of forms.
\label{hodgerema}
\end{remark}  
\indent Let $\mathcal{E}_k$ be a hermitian virtual bundle such that $ch(\mathcal{E}_k)=\alpha_k$. Suppose a holomorphic line bundle $L$ is equipped with a metric whose curvature is $F$. Consider the virtual bundles $\mathcal{L}_k$ given by the following equation.
\begin{gather}
\left [
\begin{array}{c}
\mathcal{L}_{n+1} \\
\mathcal{L}_{n} \\
\mathcal{L}_{n-1} \\
\vdots \\
\mathcal{L}_{0}
\end{array} \right ]
= N\left [ \begin{array}{ccccc} 1& 1& \frac{1}{2!} & \ldots & \frac{1}{(n+1)!} \\
 1& 2& \frac{2^{2}}{2!} & \ldots & \frac{2^{n+1}}{(n+1)!} \\
 \vdots& \vdots& \vdots & \ddots & \vdots \\
 1& (n+2)&  \frac{(n+2)^{2}}{2!} & \ldots & \frac{(n+2)^{n+1}}{(n+1)!}
 \end{array} \right ]^{-1} \left [ \begin{array}{c} L\\ L^2 \\ L^3 \\ \vdots \\ L^{n+2}\end{array} \right ],
\label{Lk}
\end{gather}
where $N$ is a large integer that clears the denominators. If $\mathcal{L}_k$ is equipped with the obvious metric, then it satisfies $ch(\mathcal{L}_k)=N \left(\frac{\sqrt{-1}F}{2\pi}\right)^{n+1-k}$. So we see that
\begin{gather}
 ch(n!N^{n} \mathcal{L}_0-\displaystyle \oplus _{k=1} ^{n} N^{n}(n+1)  \ldots (n-k+2) (n-k)\ldots 1 \mathcal{E}_k \otimes \mathcal{L}_k)\nonumber \\
 = (n+1)! N^{n+1}\left (\frac{1}{n+1}\left(\frac{\sqrt{-1}F}{2\pi}\right)^{n+1}-\sum _{k=1}^{n} \frac{1}{n-k+1}\alpha_k \left(\frac{\sqrt{-1}F}{2\pi}\right)^{n+1-k}\right).
\label{almostthere}
\end{gather}
\indent Define the hermitian virtual bundle $$\mathcal{L} =n! N^n L_0-\displaystyle \oplus _{k=1} ^{n}\frac{N^n (n+1)!}{n-k+1} \mathcal{E}_k \otimes \mathcal{L}_k$$ over $M$. Consider $\tilde{\mathcal{L}}=\pi_1^{*} \mathcal{L}$ over $M \times \mathcal{A}^{1,1}$. Define a connection $\mathbb{A}$ on $\tilde{\mathcal{L}}$ as $$\mathbb{A} (p, A) = A(p).$$ It is easy to see that this connection defines an integrable $\bar{\partial}$ operator on $\tilde{\mathcal{L}}$. We claim that
\begin{lemma}
The symplectic form $\Omega$ is the first Chern form of the Quillen metric on the Quillen determinant bundle of $\tilde{\mathcal{L}}$ (equipped with the aforementioned holomorphic structure).
\label{curvlem}
\end{lemma}
\begin{proof}
By theorem 1.27 of \cite{GS} we see that the first Chern form of the Quillen metric of the Quillen determinant of $\tilde{\mathcal{L}}$ is given by 
\begin{gather}
\tilde{\Omega} = \displaystyle \int _X [ch(\tilde{\mathcal{L}})]^{1,1}Td(X).
\label{curvform}
\end{gather}
Consider a surface full of connections in $\mathcal{A}^{1,1}$ defined by $\Phi : X\times \mathbb{R}^2 \rightarrow \mathcal{A}^{1,1}$ as $\Phi (p,x, y) = A-xa-yb$. The first Chern form  of the bundle $\Phi^{*} \pi_1 ^{*} L$ equipped with the connection $\Phi^{*}\mathbb{A}$ is  $\mathbb{F}=\omega_A +adx+bdy$. Therefore, using formulae \ref{almostthere} and \ref{curvform}  we get
\begin{gather}
\tilde{\Omega}_{x=y=0} (a,b) = \Omega_{x=y=0}(a,b),
\end{gather}
thus proving the lemma.
\end{proof}
 
\section{The almost Hitchin equations}\label{Higgssec}
\indent The usual Hitchin equations \cite{Hitchin} are
\begin{gather}
\frac{\sqrt{-1}}{2\pi}(\Theta _A + [\Phi,\Phi^{\dag}]) \wedge \omega^{n-1} = \lambda \omega^n \nonumber \\
d_A ^{0,1} \Phi = 0
.\label{usual}
\end{gather}
\indent As described in \cite{Hitchin, Dey} on Riemann surfaces, solutions to Hitchin's equations \ref{usual} are simultaneous zeroes of three moment maps $\mu_1, \mu_2, \mu_3$ corresponding to three symplectic forms $\Omega_1, \Omega_2, \Omega_3$ on an infinite-dimensional Frech\'et manifold. These symplectic forms are K\"ahler with respect to three complex structures $I, J, K$ satisfying the quaternionic relations amongst themselves. In \cite{Dey} Dey produced prequantum line bundles and connections on them whose curvatures are the aforementioned symplectic forms, i.e., she prequantised the Hitchin equations. We aim to do something similar here for the almost Hitchin equations and prove that indeed solutions exist for this new system under the assumption of Gieseker stability, thus complementing the results of \cite{Gieskerhiggs}. \\
\indent The first lemma identifies the symplectic forms $\Omega$, $\tilde{\Omega} = \Omega_2+\sqrt{-1}\Omega_3$ on $\mathcal{A}_{\mathcal{E}}^{1,1}\times \mathcal{B}$ and their moment maps that give rise to \ref{higgseq}. Actually, a part of this lemma is contained in \cite{Leung2}. We note that the tangent space of $\mathcal{A}_E^{1,1}$ at $A$ is skew-hermitian endomorphism valued $1$-forms whose $(0,1)$ part is $d_A^{0,1}$-closed (and may also be identified with $d_A ^{0,1}$-closed endomorphism valued $(0,1)$-forms) and that of $\mathcal{B}$ with endomorphism valued $(1,0)$ forms.  
\begin{lemma}
Define the symplectic form 
\begin{align*}
\Omega (\alpha \oplus \gamma, \beta \oplus \delta) &= W\displaystyle \sqrt{-1}\int _M [Tr[\alpha \wedge \beta \wedge e^{\sqrt{-1}k\omega +\frac{\sqrt{-1}}{2\pi}\Theta_A} ]_{sym}Td(X,k\omega)]^{n,n} \\ &-W n\sqrt{-1}\displaystyle \int _M Tr[\gamma \wedge \delta^{\dag}+\gamma^{\dag}\wedge \delta] \omega^{n-1},
\end{align*}
 where $sym$ indicates an anti-symmetrized product, and the complex symplectic form
$$\tilde{\Omega} (\alpha \oplus \gamma, \beta \oplus \delta) = n\displaystyle \tilde{W}\int_M Tr[-\delta \wedge \alpha + \beta \wedge \gamma] \omega^{n-1}$$ be three symplectic forms on $\mathcal{S}=\mathcal{A}_E^{1,1}\times \mathcal{B}$ (where $W$ and $\tilde{W}$ are any two positive integers). Assume that the unitary gauge group $\mathcal{G}$ acts on $\mathcal{S}$ via gauge transformations of connections and the adjoint action on $\mathcal{B}$. The moment maps corresponding to the symplectic forms are as follows. (In what follows $g$ is a skew-hermitian endomorphism.)
\begin{align*}
\mu_{A,\Phi} (g) &= W\displaystyle \int _M \Bigg (\sqrt{-1}[Tr[ge^{k\omega+\frac{\sqrt{-1}}{2\pi}\Theta_A }]Td(X,k\omega)]^{n,n} \\&+ n\frac{\sqrt{-1}}{2\pi}Tr[g[\Phi,\Phi^{\dag}]](k+l)^{n-1}\omega^{n-1} - c_k Tr[g]\omega^n \Bigg )  \nonumber \\
\tilde{\mu}_{A,\Phi} (g) &= \tilde{W}n\displaystyle \int_M tr[gd_A^{0,1} \Phi] \omega^{n-1}.
\end{align*}
\label{sympapproxhiggs}
\end{lemma}
\begin{proof}
Indeed for $\mu_A$ we have the following.  
\begin{align*}
\frac{\delta \mu_{A,\Phi}}{W} (g) &= \displaystyle \int_M [\sqrt{-1}Tr[ge^{k\omega+\frac{\sqrt{-1}}{2\pi}\Theta_A }d_A \delta A]_{sym}Td(X,k\omega)]^{n,n} \\ &+ n\frac{\sqrt{-1}}{2\pi}Tr[g[\delta\Phi,\Phi^{\dag}]+g[\Phi,\delta \Phi^{\dag}]]k^{n-1}\omega^{n-1} \nonumber \\
&= - \displaystyle \int_M \sqrt{-1}[Tr[d_A ge^{k\omega+\frac{\sqrt{-1}}{2\pi}\Theta_A } \delta A]_{sym}Td(X,k\omega)]^{n,n} \\ &+n\frac{\sqrt{-1}}{2\pi}Tr[[g,\Phi^{\dag}]\delta \Phi+[g,\Phi]\delta \Phi^{\dag}]k^{n-1}\omega^{n-1}   \nonumber \\
&=-\frac{1}{W}\Omega (d_A g\oplus [g,\Phi], \delta A\oplus\delta \Phi). 
\end{align*}
Therefore, by definition $\mu$ is the moment map corresponding to $\Omega$. Likewise for $\tilde{\mu}$ we have a similar result. (Here we consider $\mathcal{A}_E ^{1,1}$ as a complex manifold.)
\begin{gather}
\frac{\delta \tilde{\mu}_{A,\Phi}}{\tilde{W}} (g) =  n\displaystyle \int_M (tr[gd_A^{0,1} \delta\Phi]+tr[g[\delta A^{0,1},\Phi]]) \omega^{n-1} \nonumber \\
= n\displaystyle \int_M (-tr[d_A^{0,1}g \delta\Phi]+tr[[g,\Phi]\delta A^{0,1}]) \omega^{n-1}
=-\frac{1}{\tilde{W}}\tilde{\Omega} (d_A^{0,1} g\oplus [g,\Phi], \delta A\oplus\delta \Phi).
\end{gather}
\end{proof}
\indent Next we prequantise the symplectic forms akin to \cite{Dey} to produce the line bundles $\mathbf{Q}_k$ and $\mathbf{H}_1, \mathbf{H}_2$ on $\mathcal{S}$ (possibly endowed with more than one complex structure). However, unlike \cite{Dey} for higher dimensions, the non-integrability of $\bar{\partial} + \Phi^{\dag}$ poses a small problem. \\

\indent Just like in section \ref{MAsec}, let $\tilde{E}=\pi_1^{*} E$ over $X \times \mathcal{S}$. \\

\indent $\mathbf{Q}_k$  : Endow $\mathcal{S}$ with the ``usual" complex structure $I$ which is simply multiplication by $\sqrt{-1}$ on each of the factors treated as complex manifolds ($\mathcal{A}_E ^{1,1}$ being identified with $(0,1)$-form valued skew-hermitian endomorphisms) and $\mathcal{B}$ with $(1,0)$-form valued endomorphisms). Equip $\tilde{E}$ with a connection  
$$\mathbb{A} (p, A, \Phi) = A(p).$$
This connection defines an integrable $\bar{\partial}$ operator on $\tilde{E}$. Define $\mathbf{Q}_k$ to be the Quillen determinant bundle of the virtual bundle $\mathcal{E} =\tilde{E}\otimes \tilde{\mathcal{L}_2}$ where $\tilde{\mathcal{L}_2}$ is the one in section \ref{MAsec} equipped with a metric (a generalisation of the one in \cite{Dey2}) $\mathbf{h}_k = \mathbf{h}_{Quillen} \mathbf{f}$ where $\mathbf{h}_{Quillen}$ is the Quillen metric and $\mathbf{f} = \exp{(\int_M nTr(\Phi \wedge \Phi^{\dag})k^{n-1}\omega^{n-1})}$.\\
\indent The curvature of the Quillen metric is given by the families index theorem just as in the proof of lemma \ref{curvlem}. The first Chern form of the modified metric is $\Omega_{quillen} -n \frac{\spbp}{2\pi} \displaystyle \int _M Tr(\Phi\wedge\Phi^{\dag})\omega^{n-1}$ which is indeed equal to the symplectic form $\Omega$. \\

\indent $\mathbf{H_1}, \mathbf{H}_2$ : In \cite{Dey} these forms are obtained as the curvatures of the Quillen determinants of virtual bundles formed using $\tilde{E}$ equipped with the $\bar{\partial}$ operators $\bar{\partial}\pm \Phi^{\dag}$ and $\bar{\partial}\pm \sqrt{-1}\Phi^{\dag}$. However, in higher dimensions, these operators are not integrable. Despite this, one can still carry over the families index theorem and Quillen determinant construction. (After all these do arise out of Clifford connections acting on a Clifford bundle.) The issue is that it is not clear that the corresponding Quillen bundles are holomorphic. To avoid these problems we propose a very simple construction, namely, $\mathbf{H}_1$ and $\mathbf{H}_2$ are trivial line bundles. In order to describe the metrics on these bundles, we first choose an element of $\mathcal{A}_E ^{1,1}$, i.e., a connection $A_0$. Then every $\bar{\partial}$ operator is $A_0^{0,1} +a$ for some endomorphism-valued $(0,1)$-form.
\begin{gather}
\mathbf{h}_1 = \exp(\mathfrak{R}(\displaystyle \int _M nTr(\Phi \wedge a) \omega^{n-1}))\nonumber \\
\mathbf{h}_2 = \exp(\mathfrak{Im}(\displaystyle \int _M nTr(\Phi \wedge a) \omega^{n-1}))
\label{othertwometrics}
\end{gather}
It is easily seen that the first Chern forms of these metrics are precisely the two desired symplectic forms (the real and imaginary parts of the complex symplectic form in lemma \ref{sympapproxhiggs}).
\\

\indent We now proceed to prove the rest of theorem \ref{approxhiggs}. To do this, we first show that if the almost Hitchin equations have a solution, then the bundle is Gieseker-Higgs stable. This is similar to Leung's proof \cite{Leung}.
\begin{lemma}
If $(E,\Phi)$ is irreducible and there exists a solution to the almost Hitchin equations \ref{higgseq} with uniform bounds independent of $k$ for sufficiently large $k$, then for every $\Phi$-invariant coherent subsheaf $\mathcal{S}$ of $(E,\Phi)$, 
\begin{gather}
\frac{\chi(\mathcal{S} \otimes L^k)}{rk(\mathcal{S})} < \frac{\chi(E \otimes L^k)}{rk(E)}
\end{gather}
for all $k$ large enough, where $\chi(Q)$ is the holomorphic Euler characteristic of $Q$. In other words, the bundle is Gieseker-Higgs stable. 
\label{gieskerstable}
\end{lemma}
\begin{proof}
By assumption it easily follows that given an $\epsilon>0$, for large enough $k$, the almost Higgs connection satisfies $\Vert \Lambda F + \Lambda [\Phi, \Phi^{\dag}] + \Lambda c_1 (X, k\omega) - cI\Vert < \epsilon$ where $c$ is a topological constant. By conformally changing the metric we can get a new connection satisfying $\Vert \Lambda F + \Lambda [\Phi, \Phi^{\dag}] - \tilde{c}I\Vert< \epsilon$. This easily implies (just as in \cite{Leung}) that the bundle is Higgs-Mumford semi-stable. Therefore it has a Jordan-H\"older filtration by coherent subsheaves 
$$0 \subset E_1 \subset E_2 \ldots \ldots E_{l+1} = E$$
whose quotients $Q_i = \frac{E_{i}}{E_{i-1}}$ are Higgs-Mumford stable and their Mumford slopes are all equal to that of the original bundle. Moreover, $gr(E)=\oplus Q_i$ is determined uniquely up to isomorphism. With this in place, the rest of the proof is the same as the proof of proposition 3.1 in \cite{Leung}. 
\end{proof}

\indent Let $\hbar=\frac{1}{k}$. Finally, we have the following lemma that completes the proof of theorem \ref{approxhiggs}.
\begin{lemma}
If $(E,\Phi)$ is irreducible and  Higgs-Mumford semistable, then it admits an approximate almost Higgs connection.
\label{approxalmosthiggs}
\end{lemma}
\begin{proof}
  Indeed by \cite{Saini}, $(E,\Phi)$  admits a metric $h_{ap}$ whose Chern connection satisfies
\begin{gather}
\Vert \Lambda \frac{\sqrt{-1}}{2\pi}(F + [\Phi,\Phi^{\dag}]+c_1(X,\omega)) - cI\Vert < \epsilon.
\end{gather}
This by itself does not mean that that full curvature $F$ is controlled in the $L^{\infty}$ sense by $\epsilon$. Therefore we cannot use this to directly conclude the existence of an approximate almost Higgs connection. But consider the Hilbert manifold $\mathcal{P}_{p+2}(E)$ consisting of hermitian (with respect to $h_{ap}$) positive-definite endomorphisms of $E$ that are of the Sobolev class $H^{p+2}$ (where $p>>1$ to ensure that all the sections are at least $C^2$), and let $\mathcal{P}^0_{p+2}$ be a submanifold of $\mathcal{P}_{p+2}$ consisting of $g$ such that  $\displaystyle \int _M tr(g) \omega^n=1$.  Let $\mathcal{B}_{p}(E)$ consist of $H^{p}$ endomorphisms of $E$.  We shall define a map $T$ from $\mathbb{R} \times \mathcal{P}^0_{p+2,\alpha}$  to $\mathcal{B}_{p,\alpha}$ by the following. In what follows, $A$ is the connection corresponding to $h_{ap}$ and $A^g = A+g^{-1}d_A ^{1,0}g$ is the connection corresponding to $h_{ap}g$.
\begin{align}
T(\hbar,g) &=  -\frac{\sqrt{-1}}{k^{n-1}\omega^n}\Bigg ( [e^{k\omega+\frac{\sqrt{-1}}{2\pi}\Theta_{A^g} }Td(M,k\omega)]^{n,n} \nonumber \\&+ n\frac{\sqrt{-1}}{2\pi}[\Phi,\Phi^{\dag}]\wedge (k+l)^{n-1} \omega^{n-1} -c_{k,l} \omega^n \Bigg ) .
\end{align}
Note that $\Vert T(0,I) \Vert <\epsilon$. If we prove that $DT_{(0,I)}(0,h)$ is a surjective map, then by the inverse function theorem of Hilbert manifolds we see that the level set $T^{-1}(T(0,I))$ is a Hilbert submanifold where $g$ is locally a smooth function of $\hbar$, thus completing the proof. Indeed, computing the derivative we get,
\begin{gather}
DT_{(0,I)}(0,h) = \frac{n}{2\pi\omega^n }\bar{\partial} d_A ^{1,0} h \wedge \omega^{n-1} +\frac{n}{2\pi \omega^n}[\Phi, [\Phi^{\dag},h]]\omega^{n-1}.
\end{gather}
This is a self-adjoint elliptic operator and hence by the Fredholm alternative, it is enough to prove that it has no kernel. So we assume that $h$ satisfies
\begin{gather}
DT_{(0,I)}(0,h)=0.
\end{gather}
Multiplying the above equation by $\frac{1}{n}h2\pi \omega^n$, taking trace, and integrating-by-parts we obtain the following.
\begin{gather}
\displaystyle \int _M tr(d_A ^{1,0} h \bar{\partial} h) \omega^{n-1} +\int_M tr([h,\Phi][\Phi^{\dag},h]) \omega^{n-1}=0.
\end{gather}
This means that $d_A h=0$ (Note that $h$ is Hermitian and hence $d_A^{1,0}h=0$ implies that $d_A h=0$) and $[h,\Phi]=0$. So $(E,\Phi)$ decomposes into holomorphic eigenbundles of $h$. By the indecomposability assumption this means that $h=cI$ for a constant $c$ which is then forced to be zero by the normalisation assumption $\int tr(h) \omega^n =0$.
\end{proof}
\section{The Calabi-Yang-Mills equations}\label{Calabisec}
\indent In sections \ref{MAsec} and \ref{Higgssec}, the symplectic form whose moment map's zeroes are the given PDE, is clear. The difficulty is with finding the correct line bundle. In contrast, here we first come up with the prequantum line bundle and compute its curvature to get the right symplectic form.\\

\indent Consider the bundles $\tilde{E} = \pi_1^{*}E$ and $\tilde{L}=\pi_1^{*}L$ over $X \times \mathcal{N}$. Endow $\mathcal{N}$ with the same complex structure as in sections \ref{MAsec} and \ref{Higgssec}. Equip $\tilde{E}$ and $\tilde{L}$ with the connections 
$$\mathbb{A}_E(p,A_E, A_L) = A_E (p), \ and \ $$
$$\mathbb{A}_L(p,A_E, A_L) = A_L (p).$$ 
These connections (as before) define integrable $\bar{\partial}$ operators on the respective bundles. Now consider the virtual bundles $\tilde{\mathcal{L}}_0$, $\tilde{\mathcal{L}}_1$, and $\tilde{\mathcal{L}}_2$ defined in section \ref{MAsec} (possibly with different integers $N_1$, $N_2$, and $N_3$). Let $\mathcal{E} = \tilde{E}\otimes \tilde{\mathcal{L}}_2$, and $\mathcal{F} = \det(\tilde{E})\otimes \tilde{\mathcal{L}}_1$ be equipped with the induced connections.\\

\indent Define $\mathbf{Q}_{\mathcal{E}}$, $\mathbf{Q}_{\mathcal{F}}$ and $\mathbf{Q}_{\tilde{\mathcal{L}}_0}$ to be the Quillen determinant line bundles of these operators equipped with the Quillen metrics $\mathbf{h}_{E,Quillen}$, $\mathbf{h}_{F,Quillen}$, and $\mathbf{h}_{L,Quillen}$ respectively. Their first Chern forms $\Omega_{\mathcal{E}}$, $\Omega_{\mathcal{F}}$, and  $\Omega_{\tilde{\mathcal{L}}_0}$ are (as before) given by the families index theorem \cite{GS}. The form $\Omega_{\tilde{\mathcal{L}}_0}$ is almost exactly the same as in section \ref{MAsec} when all the $\alpha_k$ are zero --
\begin{gather}
\Omega_{\tilde{\mathcal{L}}} (a_E \oplus a_L,b_E \oplus b_L) =  \frac{(\sqrt{-1})^{n-1}}{(2\pi)^{n+1}}\displaystyle N_1 \int_M a_L \wedge b_L \wedge n\Theta_{L} ^{n-1},
\label{OmegaL}
\end{gather}
where $N_1$ is an integer.\\
The forms $\Omega_{\mathcal{E}}$ and $\Omega_{\mathcal{F}}$ on the other hand, are different. The reason is that $\mathcal{E}$ and  $\mathcal{F}$ depends on $L$ whereas $\tilde{\mathcal{L}}_0$ does not depend on $E$. Indeed,
\begin{align}
\Omega_{\mathcal{E}} (a_E \oplus a_L, b_E \oplus b_L) &= \displaystyle \int _M [ch(\mathcal{E})]^{1,1} (a_E \oplus a_L, b_E \oplus b_L) Td(X) \nonumber \\
&=  \displaystyle \int _M [ch(\tilde{E})ch(\tilde{\mathcal{L}}_2)]^{1,1} (a_E \oplus a_L, b_E \oplus b_L) Td(X) \nonumber 
\end{align}
which implies that
\begin{gather}
\frac{(2\pi)^{n+1}}{(\sqrt{-1})^{n-1}}\Omega_{\mathcal{E}} (a_E \oplus a_L, b_E \oplus b_L) = \displaystyle N_3\int _M tr(a_E \wedge b_E) n\Theta_{L} ^{n-1} \nonumber \\
 + N_3\int _M tr(\Theta_E^2) n{n-1 \choose 2} a_L \wedge b_L \wedge \Theta_{L} ^{n-3} \nonumber \\
+ \int _M N_3\left ( tr(\Theta_E a_E) b_L n(n-1) \Theta_{L} ^{n-2} - tr(\Theta_E b_E) a_L n(n-1)\Theta_{L}^{n-2} \right ) \nonumber 
\end{gather}
and 
\begin{align}
\Omega_{\mathcal{F}} (a_E \oplus a_L, b_E \oplus b_L) &=\displaystyle \int _M [ch(\mathcal{F})]^{1,1} (a_E \oplus a_L, b_E \oplus b_L) Td(X) \nonumber \\
&= \displaystyle \int _M [ch(\det(\tilde{E}))ch(\tilde{\mathcal{L}}_1)]^{1,1} (a_E \oplus a_L, b_E \oplus b_L) Td(X) \nonumber 
\end{align}
\begin{gather}
= N_2  \frac{(\sqrt{-1})^{n-1}}{(2\pi)^{n+1}}\int _M \Bigg ( {n \choose 2} tr(\Theta_E) \Theta_L^{n-2} a_L \wedge b_L +n\Theta_L^{n-1} tr(a_E) b_L + n\Theta_L^{n-1}a_L tr(b_E) \Bigg ).
\label{OmegaE}
\end{gather}
Now choose the integers in the definition of the virtual bundles to be so large that $\mathbf{Q}_{\mathcal{E}}^{\alpha}$ and $\mathbf{Q}_{\mathcal{F}}^{-\alpha \lambda}$ make sense, and choose potentially different $\tilde{\mathcal{L}}$s so that $N_1=N_2=N_3=N$. Then the bundle $\mathbf{Q}_{\alpha}=\mathbf{Q}_{\mathcal{E}}^{-\alpha}\otimes \mathbf{Q}_{\tilde{\mathcal{L}}_0} \otimes \mathbf{Q}_{\mathcal{F}}^{\alpha \lambda}$ is the desired bundle on $\mathcal{N}$ whose curvature form $\Omega_{\alpha}$ is
\begin{gather}
\Omega_{\alpha} = \Omega_{\tilde{\mathcal{L}}_0} - \alpha \Omega_{\mathcal{E}} - \lambda \alpha \Omega_{\mathcal{F}}\nonumber 
\end{gather}
meaning that
\begin{gather}
\frac{(2\pi)^{n+1}}{(\sqrt{-1})^{n-1}}\Omega_{\alpha} = -N\alpha \Bigg ( \displaystyle \int _M tr(a_E \wedge b_E) n\Theta_{L} ^{n-1} \nonumber \\ 
+  \int _M \left ( tr(\Theta_E a_E) b_L n(n-1) \Theta_{L} ^{n-2} + a_Ltr(\Theta_E b_E)  n(n-1)\Theta_{L}^{n-2} \right ) \nonumber \\
+\lambda \int_M n\Theta_L^{n-1} tr(a_E) b_L + \lambda \int_M n\Theta_L^{n-1}a_L tr(b_E) \Bigg ) \nonumber \\
+ N \Bigg (-\alpha \int _M tr(\Theta_E^2) n{n-1 \choose 2} a_L \wedge b_L \wedge \Theta_{L} ^{n-3}-\lambda \alpha \int_M {n \choose 2} tr(\Theta_E) \Theta_L^{n-2} a_L \wedge b_L \nonumber \\ 
+ \int_M a_L \wedge b_L \wedge n\Theta_{L} ^{n-1}  \Bigg ).
\label{Omegaal}
\end{gather}
Now we compute the moment map under the action of the gauge group.
\begin{lemma}
The moment map $\mu$ under the action of the unitary gauge group $\mathcal{G}$ is the following.
\begin{gather}
\frac{(2\pi)^{n+1}}{(\sqrt{-1})^{n-1}}\mu_{A_E, A_L} (g_E, g_L) =  -N\alpha  \displaystyle \int_M \Bigg ( tr(g_E \Theta_E)n\Theta_L ^{n-1}+\lambda tr(g_E) \Theta_L^n  \Bigg ) \nonumber \\
+ N\displaystyle \int_M g_L \Bigg (\Theta_L^n -\frac{(2\pi)^{n}}{(\sqrt{-1})^{n}} \eta-\alpha tr(\Theta_E^2) \frac{n(n-1)}{2} \Theta_L^{n-2} -\lambda \alpha tr(\Theta_E) \frac{n}{2} \Theta_L^{n-1} \Bigg ).
\label{momentmapexp}
\end{gather}
\label{momentmapundergauge}
\end{lemma}
\begin{proof}
Indeed, computing the variation of $\mu$ along $\delta A_E \oplus \delta A_L$ we get the following expression.
\begin{gather}
\frac{(2\pi)^{n+1}}{(\sqrt{-1})^{n-1}}\delta \mu_{A_E, A_L} (g_E, g_L) = -N\alpha  \displaystyle \int_M \Bigg ( tr(g_E d_{A_E} \delta A_E)n\Theta_L ^{n-1}\nonumber \\
+ tr(g_E \Theta_E) n(n-1)\Theta_L^{n-2} d_{A_L} \delta A_L
+  \lambda tr(g_E) n\Theta_L^{n-1}d_{A_L} \delta A_L  \Bigg )\nonumber \\ 
+ N\displaystyle \int_M g_L \Bigg (n\Theta_L^{n-1}d_{A_L}\delta A_L 
-\alpha tr(2\Theta_E d_{A_E} \delta A_E) \frac{n(n-1)}{2} \Theta_L^{n-2}\nonumber \\ - \alpha tr(\Theta_E^2) \frac{n(n-1)(n-2)}{2} \Theta_L^{n-3} d_{A_L} \delta A_L 
 -\lambda \alpha tr(d_{A_E} \delta A_E) \frac{n}{2} \Theta_L^{n-1} \nonumber \\ - \lambda \alpha tr(\Theta_E) \frac{n(n-1)}{2} \Theta_L^{n-2}d_{A_L}\delta A_L\Bigg ) 
= - \frac{(2\pi)^{n+1}}{(\sqrt{-1})^{n-1}}\Omega_{\alpha} (d_{A_E} g_E \oplus d_{A_L}g_L, \delta A_E \oplus \delta A_L ),
\end{gather}
where the last equality is obtained by integration-by-parts. Therefore by definition, indeed $\mu$ is the moment map.
\end{proof}
Lemma \ref{momentmapundergauge} implies that the zeroes of the moment map precisely correspond to the solutions of the Calabi-Yang-Mills equations. \\
\indent The only thing left to complete the proof of theorem \ref{calabiyangmills} is to prove that for small $\alpha$ there exists a solution. This is indeed an easy implicit function theorem argument perturbing around $\alpha=0$ where there is a solution thanks to the Calabi conjecture and the Donaldson-Uhlenbeck-Yau theorem. \qed

\end{document}